\documentclass[letterpaper,12pt]{article}
\usepackage[backend=biber,style=numeric,sorting=nyt,url=false,doi=false,isbn=false,eprint=false]{biblatex}
\usepackage{graphicx}
\usepackage{palatino,mathrsfs}
\usepackage[mathcal]{euler}
\usepackage{xcolor}
\usepackage{epstopdf,epsfig}
\usepackage{pdfpages}
\usepackage{comment}
\usepackage{amsmath,amsthm,amsfonts,amssymb,latexsym,amscd,enumerate,url,hyperref}
\usepackage{mathtools}
\usepackage[all,knot]{xy}
\xyoption{arc}
\usepackage{multirow}
\usepackage{caption}
\usepackage[new]{old-arrows}
\usepackage{nicematrix}
\usepackage{bm}
\usepackage[T1]{fontenc}
\usepackage[utf8]{inputenc}
\usepackage{microtype}

\newtheorem{theorem}{Theorem}[section]
\newtheorem*{theorem*}{Theorem}

\newtheorem*{corollary*}{Corollary}
\newtheorem{lemma}{Lemma}[section]
\newtheorem*{lemma*}{Lemma}
\newtheorem{proposition}{Proposition}[section]
\newtheorem*{proposition*}{Proposition}

\newtheorem*{claim*}{Claim}
\theoremstyle{definition}
\newtheorem{definition}{Definition}[section]
\newtheorem*{definition*}{Definition}
\newtheorem{example}{Example}[section]
\newtheorem*{example*}{Example}

\newtheorem*{question*}{Question}

\newtheorem*{problem*}{Problem}

\newtheorem*{remark*}{Remark}

\def\R{\mathbb R}
\def\C{\mathbb C}
\def\N{\mathbb N}
\def\Z{\mathbb Z}

\def\fg{\mathfrak{g}}
\def\fk{\mathfrak{k}}

\def\fl{\mathfrak{l}}
\def\fp{\mathfrak{p}}

\def\g{\boldsymbol{\mathfrak{g}}}

\def\A{\mathbb{A}}
\def\fgfam{\boldsymbol{\mathfrak{g}}}
\def\flfam{\boldsymbol{\mathfrak{l}}}

\DeclareMathOperator{\Ad}{Ad}
\DeclareMathOperator{\diag}{diag}
\DeclareMathOperator{\GL}{GL}

\DeclareMathOperator{\Sub}{Sub}
\DeclareMathOperator{\im}{Im}
\DeclareMathOperator{\pr}{Pr}

\title{Dual contractions and algebraic families}
\author{Eyal Subag\thanks{Department of Mathematics, Bar-Ilan University, Ramat-Gan, 5290002 Israel.}}
\date{}

\begin{document}

\maketitle
\begin{abstract}
We introduce a duality  for In\"on\"u-Wigner contractions attached to real symmetric Lie algebras. Starting from a symmetric pair $(\fg,\theta)$, we define a dual real form $\fg^{*}$ inside the complexification of $\fg$ and consider the corresponding contraction with respect to the common fixed-point subalgebra $\fg^{\theta}$. The main result shows that the original contraction and its dual appear as real fibers of a single algebraic family of complex Lie algebras equipped with an anti-holomorphic involution. This places the two contractions in one geometric framework and connects them with the algebraic-family methods developed in recent work on contractions, real forms, and hidden symmetries.
\end{abstract}

\section{Introduction}
Contractions of Lie algebras provide a systematic way to pass from one symmetry algebra to another by a limiting procedure. A classical example in physics is the passage from the Poincar\'e algebra to the Galilei algebra in the non-relativistic limit. A closely related example, and one that is more directly relevant for this paper, is that both $\mathfrak{so}(4)$ and $\mathfrak{so}(3,1)$ contract with respect to their common symmetric subalgebra $\mathfrak{so}(3)$ to the Euclidean algebra $\mathfrak{iso}(3)$; see Chapter~10 of Ref.~\cite{MR1275599} and also Refs.~\cite{MR779059,SBBM}. This phenomenon is not isolated: many classical contractions arise from pairs of Lie algebras sharing a common subalgebra; see, for example, \cite[Ch. 10]{MR1275599}.

A particularly important class of contractions is formed by the In\"on\"u-Wigner contractions introduced in Ref.~\cite{MR55352}; see also Ref.~\cite{Saletan61}. These contractions are associated with a decomposition $\fg=\fk\oplus\fp$ in which $\fk$ is a Lie subalgebra. The resulting contracted algebra is the semidirect product $\fk\ltimes \fp$, where $\fp$ becomes an abelian ideal. When $\fk$ is the fixed-point subalgebra $\fg^{\theta}$ of an involution $\theta$ of $\fg$, the pair $(\fg,\theta)$ is a symmetric Lie algebra, and the contraction takes the form
\begin{equation}
\fg^{\theta}\ltimes \fg^{-\theta}.
\end{equation}

Motivated in part by the semisimple duality appearing in \cite{FLENSTEDJENSEN1978106}, we introduce a dual real symmetric Lie algebra 
\begin{equation}
\fg^{*}=\fg^{\theta}\oplus i\fg^{-\theta}
\end{equation}  inside the complexification  of $\fg$.

We call the In\"on\"u-Wigner contraction of $\fg^{*}$ with respect to the same fixed-point subalgebra $\fg^{\theta}$ the \emph{dual contraction} of the original one.




Our main theorem shows that the original contraction and its dual are naturally linked by an algebraic family of Lie algebras $\g$. More precisely, there is a real structure of $\g$ such that  the fibers of the corresponding family of real Lie algebras $\g^{\boldsymbol{\sigma}}$ satisfy  
\[\g^{\boldsymbol{\sigma}}|_t\cong \begin{cases}
 \fg^*,&  t<0 \\
 \fg^{\theta}\ltimes \fg^{-\theta},& t=0\\
  \fg,& t>0.
\end{cases}\]

Algebraic families of Lie algebras and Harish-Chandra pairs were introduced and studied in Refs.~\cite{eyalbern,MR4123111,MR3797197}. They provide a natural framework for relating different real forms of a complex Lie algebra as well as their contractions. One of the motivations for the present work comes from the hidden-symmetry families associated with the hydrogen atom, studied in Refs.~\cite{MR3827131,MR4460278}.
In those papers, a natural algebraic family of complex Lie algebras was constructed, together with a real structure whose real fibers satisfy 
\[\g^{\boldsymbol{\sigma}}|_E\cong \begin{cases}
 \mathfrak{so}(n+1),&  E<0 \\
 \mathfrak{so}(n)\ltimes \R^n,& E=0\\
  \mathfrak{so}(n,1),& E>0.
\end{cases}\]
In particular, the family  encodes the usual  hidden symmetries of the hydrogen atom and simultaneously interpolates the dual contractions
\begin{equation}
\mathfrak{so}(n+1)\longrightarrow \mathfrak{so}(n)\ltimes \R^{n}
\qquad \text{and} \qquad
\mathfrak{so}(n,1)\longrightarrow \mathfrak{so}(n)\ltimes \R^{n}.
\end{equation}

The paper is organized as follows. In Sec.~\ref{sec:contractions} we review contractions of Lie algebras and the relation between simple In\"on\"u-Wigner contractions and subalgebras. In Sec.~\ref{sec:families} we recall the notion of an algebraic family of Lie algebras over the affine line. In Sec.~\ref{sec:dual-contraction} we define the dual symmetric Lie algebra and the dual contraction. In Sec.~\ref{sec:main} we prove that the two dual contractions occur as fibers of a single algebraic family, and we recall an explicit realization of that family. Appendix~\ref{sec:appendix} contains two auxiliary results used in Sec.~\ref{sec:contractions} and an embedding lemma used in Sec.~\ref{subsec:explicit}.

This research was supported by the Israel Science Foundation (grant No. 1040/22).

\section{Contractions}\label{sec:contractions}
\subsection{General contractions}
Let $\fg=(V,[\cdot,\cdot])$ be a real or complex Lie algebra and let $T\in \GL(V)$. The formula
\begin{equation}
[X,Y]_{T}:=T^{-1}([T(X),T(Y)]), \qquad X,Y\in V,
\end{equation}
defines another Lie bracket on $V$, and the Lie algebra $\fg_{T}:=(V,[\cdot,\cdot]_{T})$ is isomorphic to $\fg$.

When the linear map $T$ is replaced by a family of invertible maps and one passes to a limit, the resulting Lie algebra need not be isomorphic to the original one.

\begin{definition}
Let $\fg=(V,[\cdot,\cdot])$ be a real or complex Lie algebra and let $\{T_{\epsilon}\}_{\epsilon>0}$ be a family of invertible linear maps on $V$ such that for every $X,Y\in V$ the limit
\begin{equation}
[X,Y]_{T(0)}:=\lim_{\epsilon\to 0^{+}}T_{\epsilon}^{-1}([T_{\epsilon}(X),T_{\epsilon}(Y)])
\end{equation}
exists. Then $[\cdot,\cdot]_{T(0)}$ is a Lie bracket on $V$, and the Lie algebra
\begin{equation}
\fg_{0}^{T}:=(V,[\cdot,\cdot]_{T(0)})
\end{equation}
is called the \emph{contraction} of $\fg$ with respect to $\{T_{\epsilon}\}_{\epsilon>0}$. We denote this by
\begin{equation}
\fg \xrightarrow[]{T_{\epsilon}} \fg_{0}^{T}.
\end{equation}
\end{definition}

When the dependence of $T_{\epsilon}$ on $\epsilon$ is regular (for example, continuous, differentiable, or algebraic), the family of Lie algebras $\{\fg_{\epsilon}\}_{\epsilon>0}$ with brackets $[\cdot,\cdot]_{T(\epsilon)}$ provides a controlled approximation of the limiting algebra $\fg_{0}^{T}$.

\subsection{Generalized In\"on\"u-Wigner contractions}
A contraction of a finite-dimensional Lie algebra
\begin{equation}
\fg \xrightarrow[]{T_{\epsilon}} \fg_{0}^{T}
\end{equation}
is called a \emph{generalized In\"on\"u-Wigner contraction} if there is a basis $E$ of $V$ such that, for every $\epsilon>0$, the matrix of $T_{\epsilon}$ in the basis $E$ is diagonal of the form
\begin{equation}
[T_{\epsilon}]_{E}=\diag(\epsilon^{n_{1}},\epsilon^{n_{2}},\ldots,\epsilon^{n_{\dim V}})
\end{equation}
for some integers $n_{j}\in \Z$. When all exponents belong to $\{0,1\}$, the contraction is called a \emph{simple In\"on\"u-Wigner contraction}, or simply an \emph{In\"on\"u-Wigner contraction}.

For a fixed Lie algebra $\fg$, two contractions
\begin{equation}
\fg \xrightarrow[]{T_{\epsilon}} \fg_{0}^{T}
\qquad \text{and} \qquad
\fg \xrightarrow[]{S_{\epsilon}} \fg_{0}^{S}
\end{equation}
are called \emph{equivalent} if the Lie algebras $\fg_{0}^{T}$ and $\fg_{0}^{S}$ are isomorphic.

Generalized In\"on\"u-Wigner contractions form a large and important class of contractions \cite{MR1809459}, although they are not universal in complete generality; see Ref.~\cite{Popovych2010} for a counterexample.

\subsection{Simple In\"on\"u-Wigner contractions and subalgebras}
Despite the previous remark, simple In\"on\"u-Wigner contractions are particularly important. Most contractions that arise in classical mathematical physics are of this type, and they admit a direct geometric interpretation in terms of subalgebras.

Suppose that
\begin{equation}
\fg \xrightarrow[]{T_{\epsilon}} \fg_{0}^{T}.
\end{equation}
Then for every $v\in V$ the limit
\begin{equation}
T_{0}(v):=\lim_{\epsilon\to 0^{+}}T_{\epsilon}(v)
\end{equation}
exists and defines a projection operator on $V$. Moreover, $\im(T_{0})$ is a Lie subalgebra of both $\fg$ and $\fg_{0}^{T}$, while $\ker(T_{0})$ is an abelian ideal of $\fg_{0}^{T}$. In particular,
\begin{equation}
\fg_{0}^{T}=\im(T_{0})\ltimes \ker(T_{0}).
\end{equation}
These observations go back to the original paper of In\"on\"u and Wigner \cite[Theorem~1]{MR55352}.

Conversely, let $\fk$ be a Lie subalgebra of $\fg$, and choose a vector-space complement $\fp$ so that $\fg=\fk\oplus \fp$. Denote by $\pr_{\fk}$ and $\pr_{\fp}$ the associated projections. The linear maps
\begin{equation}
T_{\epsilon}^{\fk,\fp}=\pr_{\fk}+\epsilon\pr_{\fp}, \qquad \epsilon>0,
\end{equation}
define a simple In\"on\"u-Wigner contraction of $\fg$, whose limit is
\begin{equation}
\fg_{0}^{\fk,\fp}=\fk\ltimes \fp.
\end{equation}
Here $\fp$ is an abelian ideal and, for $X\in \fk$ and $Y\in \fp$, the mixed bracket is
\begin{equation}
[X,Y]_{T^{\fk,\fp}(0)}=\pr_{\fp}([X,Y]).
\end{equation}

The choice of the complement $\fp$ does not affect the isomorphism class of the contracted Lie algebra.

\begin{lemma}\label{lem:complements-jmp}
Let $\fg$ be a real or complex Lie algebra and let $\fk$ be a subalgebra of $\fg$. If $\fp$ and $\fp'$ are vector-space complements of $\fk$ in $\fg$, then
\begin{equation}
\fg_{0}^{\fk,\fp}\cong \fg_{0}^{\fk,\fp'}.
\end{equation}
\end{lemma}

\begin{proof}
Treat the quotient vector space $\fg/\fk$ as an abelian Lie algebra. The adjoint action $\pi$ of $\fk$ on $\fg/\fk$ is given by
\begin{equation}
\pi(X)(Y+\fk)=[X,Y]+\fk, \qquad \forall X\in \fk,\ Y\in \fg.
\end{equation}
This defines a semidirect product $\fk\ltimes (\fg/\fk)$ that depends only on $\fk$. For each complement $\fp$, the map
\begin{equation}
\psi_{\fp}:\fk\ltimes \fp \longrightarrow \fk\ltimes (\fg/\fk),
\qquad
\psi_{\fp}(X,Y)=(X,Y+\fk),
\end{equation}
is an isomorphism of Lie algebras. Hence both $\fg_{0}^{\fk,\fp}$ and $\fg_{0}^{\fk,\fp'}$ are canonically isomorphic to $\fk\ltimes (\fg/\fk)$.
\end{proof}

The preceding lemma shows that every subalgebra $\fk$ gives rise to a well-defined equivalence class of In\"on\"u-Wigner contractions. However, this dependence on $\fk$ is subtler than is sometimes suggested in the literature.

\begin{example}
The Lie algebra $\mathfrak{so}(4)$ contains two isomorphic ideals $\fk_{1}$ and $\fk_{2}$, each isomorphic to $\mathfrak{so}(3)$, such that
\begin{equation}
\mathfrak{so}(4)=\fk_{1}\oplus \fk_{2}.
\end{equation}
The contraction associated with the decomposition $\mathfrak{so}(4)=\fk_{1}\oplus \fk_{2}$ is
\begin{equation}
\fg_{0}^{\fk_{1},\fk_{2}}=\fk_{1}\ltimes \fk_{2},
\end{equation}
and in this case the action of $\fk_{1}$ on $\fk_{2}$ is trivial, so the contraction is simply the direct sum of the ideal $\fk_{1}$ with the abelian ideal $\fk_{2}$, and $\fg_{0}^{\fk_{1},\fk_{2}}\cong \mathfrak{so}(3)\oplus \R^3$.

Now choose an isomorphism $\psi:\fk_{1}\to \fk_{2}$ and set
\begin{equation}
\fk:=\{X+\psi(X):X\in \fk_{1}\},
\qquad
\fp:=\{X-\psi(X):X\in \fk_{1}\}.
\end{equation}
Then $\fk$ is again a Lie subalgebra isomorphic to $\mathfrak{so}(3)$, whereas $\fp$ is not a Lie subalgebra. As a $\fk$-module, however, $\fp$ is isomorphic to the adjoint representation of $\fk$. Clearly, as vector spaces  
\begin{equation}
\fg=\fk\oplus \fp.
\end{equation}
The corresponding contraction
\begin{equation}
\fg_{0}^{\fk,\fp}=\fk\ltimes \fp
\end{equation}
is therefore isomorphic to the Lie algebra $\mathfrak{iso}(3)=\mathfrak{so}(3)\ltimes \R^3$ of the Euclidean group, which is not isomorphic to $\fg_{0}^{\fk_{1},\fk_{2}}$.
\end{example}

To capture the correct invariant relevant to contractions, we introduce an equivalence relation on subalgebras.

\begin{definition}
Let $\fg$ be a finite-dimensional Lie algebra. Two subalgebras $\fk_{1},\fk_{2}\in \Sub(\fg)$ are called \emph{equivalent}, and we write $\fk_{1}\sim \fk_{2}$, if the Lie algebras
\begin{equation}
\fk_{1}\ltimes (\fg/\fk_{1})
\qquad \text{and} \qquad
\fk_{2}\ltimes (\fg/\fk_{2})
\end{equation}
are isomorphic. In both cases the quotient is viewed as an abelian Lie algebra equipped with the natural adjoint action of the corresponding subalgebra.
\end{definition}

The appendix contains a useful sufficient condition for two subalgebras to be equivalent; see Proposition~\ref{prop:equivalent-subalgebras-jmp}.

We denote by $\operatorname{IWC}(\fg)$ the set of equivalence classes of simple In\"on\"u-Wigner contractions of $\fg$.

\begin{theorem}\label{thm:iwc-bijection-jmp}
Let $\fg$ be a finite-dimensional Lie algebra. The assignment that sends a subalgebra $\fk\subseteq \fg$ to the equivalence class of the contraction $\fg_{0}^{\fk,\fp}$ induces a bijection
\begin{equation}
\Sub(\fg)/\sim\;\xrightarrow{\ \sim\ }\; \operatorname{IWC}(\fg).
\end{equation}
\end{theorem}

\begin{proof}
By Lemma~\ref{lem:complements-jmp}, the isomorphism class of $\fg_{0}^{\fk,\fp}$ does not depend on the chosen complement $\fp$, so the assignment is well defined. Surjectivity is immediate from the construction of simple In\"on\"u-Wigner contractions. Injectivity is precisely the definition of the equivalence relation $\sim$.
\end{proof}

\section{Algebraic families of Lie algebras}\label{sec:families}
\subsection{Complex algebraic families}
An algebraic family of complex Lie algebras over a complex algebraic variety is, roughly speaking, a family of Lie algebras that varies algebraically with the parameter. For general definitions and examples we refer to Refs.~\cite{eyalbern,MR4123111}.

In this paper we work only over the complex affine line $\A^{1}_{\C}$. In that case the definition becomes especially concrete.

\begin{definition}
An \emph{algebraic family of complex Lie algebras over $\A^{1}_{\C}$} is a finite-rank free $\C[z]$-module $\fgfam$ equipped with the structure of a Lie algebra over the ring $\C[z]$.
\end{definition}

For $\alpha\in \C$, let $I_{\alpha}$ be the ideal $(z-\alpha)\C[z]$. By definition  \emph{the fiber of $\fgfam$ over $\alpha$} is the complex Lie algebra
\begin{equation}
\fgfam|_{\alpha}:=\fgfam/I_{\alpha}\fgfam,
\end{equation}
where we use the isomorphism $\C[z]/I_{\alpha}\cong \C$.

If $\{e_{1},\ldots,e_{n}\}$ is a basis of $\fgfam$ over $\C[z]$, so that $\fgfam=\bigoplus_{i=1}^{n}\C[z]e_{i}$, then the structure constants are polynomials $C_{ij}^{k}(z)\in \C[z]$ defined by
\begin{equation}
[e_{i},e_{j}]=\sum_{k=1}^{n}C_{ij}^{k}(z)e_{k}.
\end{equation}
After passing to the fiber at $\alpha$, the induced structure constants are the complex numbers $C_{ij}^{k}(\alpha)$.

\begin{example}
The $\C[z]$-module $\mathfrak{sl}_{3}(\C[z])$ of traceless $3\times 3$ matrices with entries in $\C[z]$ is an algebraic family of complex Lie algebras over $\A^{1}_{\C}$, with Lie bracket given by the commutator.

For each $\alpha\in \C$, evaluation at $z=\alpha$ defines an isomorphism
\begin{equation}
\mathfrak{sl}_{3}(\C[z])|_{\alpha}\cong \mathfrak{sl}_{3}(\C).
\end{equation}
Thus this is a constant family, isomorphic to $\C[z]\otimes_{\C}\mathfrak{sl}_{3}(\C)$.

Now consider the subfamily
\begin{equation}
\fgfam:=\left\{\begin{pmatrix}
0 & a(z) & b(z)\\
-a(z) & 0 & c(z)\\
zb(z) & zc(z) & 0
\end{pmatrix}\middle| a,b,c\in \C[z]\right\}.
\end{equation}
Then for each $\alpha\in \C$,
\begin{equation}
\fgfam|_{\alpha}\cong \left\{\begin{pmatrix}
0 & a & b\\
-a & 0 & c\\
\alpha b & \alpha c & 0
\end{pmatrix}\middle| a,b,c\in \C\right\}
\cong
\begin{cases}
\mathfrak{so}(3,\C), & \alpha\neq 0,\\
\mathfrak{so}(2,\C)\ltimes \C^{2}, & \alpha=0.
\end{cases}
\end{equation}
\end{example}

\subsection{Real algebraic families}
\begin{definition}
An \emph{algebraic family of real Lie algebras over $\A^{1}_{\R}$} is a finite-rank free $\R[z]$-module $\flfam$ equipped with the structure of a Lie algebra over $\R[z]$.
\end{definition}

For $\alpha\in \R$, let $I_{\alpha}(\R):=(z-\alpha)\R[z]$. The fiber of $\flfam$ over $\alpha$ is
\begin{equation}
\flfam|_{\alpha}:=\flfam/I_{\alpha}(\R)\flfam.
\end{equation}

Any algebraic family of complex Lie algebras over $\A^{1}_{\C}$ becomes, by restriction of scalars, an algebraic family of real Lie algebras over $\A^{1}_{\R}$.

Recall that an anti-holomorphic involution of a complex Lie algebra $\fg$ is an automorphism $\sigma$ of $\fg$ viewed as a real Lie algebra such that
\begin{equation}
\sigma(\alpha X)=\overline{\alpha}\,\sigma(X),
\qquad \forall \alpha\in \C,
\ X\in \fg,
\end{equation}
and $\sigma^{2}=\mathrm{id}$.

\begin{definition}\label{d3}
Let $\fgfam$ be an algebraic family of complex Lie algebras over $\A^{1}_{\C}$. An \emph{anti-holomorphic involution} of $\fgfam$ is a morphism of Lie algebras over $\R[z]$
\begin{equation}
\boldsymbol{\sigma}:\fgfam\longrightarrow \fgfam
\end{equation}
such that $\boldsymbol{\sigma}^{2}=\mathrm{id}$, and for all $X\in \fgfam$ and $f(z)\in \C[z]$,
 
\begin{equation}
\boldsymbol{\sigma}(f(z)X)=\overline{f}(z)\,\boldsymbol{\sigma}(X),
\end{equation}
 
where if $f(z)=\sum_{j}a_{j}z^{j}$ then $\overline{f}(z):=\sum_{j}\overline{a}_{j}z^{j}$.
\end{definition}

\begin{lemma}
Let $\boldsymbol{\sigma}:\fgfam\to \fgfam$ be an anti-holomorphic involution of an algebraic family of complex Lie algebras over $\A^{1}_{\C}$. Then
\begin{equation}
\fgfam^{\boldsymbol{\sigma}}:=\{X\in \fgfam: \boldsymbol{\sigma}(X)=X\}
\end{equation}
is an algebraic family of real Lie algebras over $\A^{1}_{\R}$.
\end{lemma}

\begin{proof}
The fixed-point set $\fgfam^{\boldsymbol{\sigma}}$ is stable under the Lie bracket and under
multiplication by $\R[z]$, hence it is an $\R[z]$-Lie subalgebra of $\fgfam$.
Since $\g$ is free as an $\R[z]$-module and $\R[z]$ is a principal ideal
domain, every $\R[z]$-submodule of $\g$ is free. Therefore $\fgfam^{\boldsymbol{\sigma}}$
is a free $\R[z]$-module. 
\end{proof}

\begin{example}
For the family
\begin{equation}
\fgfam:=\left\{\begin{pmatrix}
0 & a(z) & b(z)\\
-a(z) & 0 & c(z)\\
zb(z) & zc(z) & 0
\end{pmatrix}\middle| a,b,c\in \C[z]\right\},
\end{equation}
define $\boldsymbol{\sigma}:\fgfam\to \fgfam$ by entrywise complex conjugation. Then, for $\alpha\in \R$,
\begin{equation}
\fgfam^{\boldsymbol{\sigma}}|_{\alpha}
\cong
\left\{\begin{pmatrix}
0 & a & b\\
-a & 0 & c\\
\alpha b & \alpha c & 0
\end{pmatrix}\middle| a,b,c\in \R\right\}
\cong
\begin{cases}
\mathfrak{so}(3), & \alpha<0,\\
\mathfrak{so}(2)\ltimes \R^{2}, & \alpha=0,\\
\mathfrak{so}(2,1), & \alpha>0.
\end{cases}
\end{equation}
\end{example}

\section{The dual contraction}\label{sec:dual-contraction}
\subsection{Symmetric Lie algebras}
A \emph{symmetric Lie algebra} is a pair $(\fg,\theta)$ in which $\fg$ is a Lie algebra and $\theta$ is an involutive automorphism of $\fg$. Then $\fg$ is naturally $\Z_{2}$-graded:
\begin{equation}
\fg=\fg^{\theta}\oplus \fg^{-\theta},
\end{equation}
where
\begin{equation}
\fg^{\pm\theta}:=\{X\in \fg: \theta(X)=\pm X\}.
\end{equation}
The bracket respects this grading:
\begin{subequations}
\begin{align}
[\fg^{\theta},\fg^{\theta}]&\subseteq \fg^{\theta},\\
[\fg^{\theta},\fg^{-\theta}]&\subseteq \fg^{-\theta},\\
[\fg^{-\theta},\fg^{-\theta}]&\subseteq \fg^{\theta}.
\end{align}
\end{subequations}
The fixed-point subalgebra $\fg^{\theta}$ is called a \emph{symmetric subalgebra}.

\begin{example}
Let $p,q,d\in \N_{0}$ with $d>0$, and define
\begin{equation}
J_{p,q}:=\diag(\underbrace{1,\ldots,1}_{p\text{ times}},\underbrace{-1,\ldots,-1}_{q\text{ times}}),
\end{equation}
\begin{equation}
J_{p,d,q}:=\diag(\underbrace{1,\ldots,1}_{p\text{ times}},\underbrace{-1,\ldots,-1}_{d\text{ times}},\underbrace{1,\ldots,1}_{q\text{ times}}).
\end{equation}
Then
\begin{equation}
\fg=\mathfrak{so}(p+d,q):=\{X\in \mathfrak{gl}_{p+d+q}(\R): X^{t}J_{p+d,q}=-J_{p+d,q}X\}
\end{equation}
together with the involution $\theta_{p,d,q}:=\Ad(J_{p,d,q})$ is a real symmetric Lie algebra.
\end{example}

\subsection{The dual symmetric Lie algebra}\label{subsec:dual-symmetric}
Let $(\fg,\theta)$ be a finite-dimensional real symmetric Lie algebra. Extend $\theta$ complex linearly to an involution $\widetilde{\theta}$ of the complexification
\begin{equation}
\fg(\C):=\C\otimes_{\R}\fg.
\end{equation}
We identify $\fg$ with its natural image in $\fg(\C)$. The real form $\fg\subseteq \fg(\C)$ determines an anti-holomorphic involution
\begin{equation}
\sigma:\fg(\C)\longrightarrow \fg(\C)
\end{equation}
that acts as the identity on $\fg$.

\begin{lemma}\label{lem:dual-symmetric-jmp}
Let $(\fg,\theta)$ be a finite-dimensional real symmetric Lie algebra. Define the anti-holomorphic involution
\begin{equation}
\sigma^{*}:=\sigma\widetilde{\theta}
\end{equation}
of $\fg(\C)$, and let
\begin{equation}
\fg^{*}:=\fg(\C)^{\sigma^{*}}
\end{equation}
be the corresponding real form. Then:
\begin{enumerate}
    \item The involutions $\widetilde{\theta}$, $\sigma$, and $\sigma^{*}$ commute.
    \item The pair $(\fg^{*},\widetilde{\theta}|_{\fg^{*}})$ is a real symmetric Lie algebra.
    \item The eigenspace decomposition of $\fg^{*}$ with respect to $\widetilde{\theta}|_{\fg^{*}}$ is
    \begin{equation}
    \fg^{*}=(\R\otimes_{\R}\fg^{\theta})\oplus (i\R\otimes_{\R}\fg^{-\theta}).
    \end{equation}
    \item If the same construction is applied to $(\fg^{*},\widetilde{\theta}|_{\fg^{*}})$, then one recovers $(\fg,\theta)$.
\end{enumerate}
\end{lemma}

\begin{proof}
Since $\widetilde{\theta}$  
is the complexification of  the $\R$-linear involution $\theta:\fg\longrightarrow \fg$, it preserves the real form $\fg\subset \fg(\C)$, and therefore commutes with $\sigma$. 
 Hence $\sigma^{*}=\sigma\widetilde{\theta}$ is again an anti-holomorphic involution and all three involutions commute. The fixed-point space $\fg^{*}=\fg(\C)^{\sigma^{*}}$ is therefore stable under $\widetilde{\theta}$, proving items~(1) and~(2).

Write $X\in \fg(\C)$ as $X=X_{+}+X_{-}$ with $X_{\pm}\in \fg(\C)^{\pm\widetilde{\theta}}$. Then $X\in \fg^{*}$ if and only if
\begin{equation}
\sigma(X_{+})=X_{+}
\qquad \text{and} \qquad
\sigma(X_{-})=-X_{-}.
\end{equation}
Equivalently, $X_{+}\in \fg^{\theta}$ and $X_{-}\in i\fg^{-\theta}$. This proves item~(3). Finally, applying the same construction once more replaces $i\fg^{-\theta}$ by $\fg^{-\theta}$ and therefore recovers the original real form, proving item~(4).
\end{proof}

\begin{definition}[Compare with Sec. 4 of ~\cite{FLENSTEDJENSEN1978106}]
The symmetric Lie algebra $(\fg^{*},\widetilde{\theta}|_{\fg^{*}})$ of Lemma~\ref{lem:dual-symmetric-jmp} is called the \emph{dual symmetric Lie algebra} of $(\fg,\theta)$.
\end{definition}

We remark that in the semisimple setting, this duality  appeared already in the  work of
Flensted-Jensen ~\cite{FLENSTEDJENSEN1978106}, and  is also closely related
to what is often called Cartan duality (or c-duality); see, for example, ~\cite{MR1834454,BABA2021101751,MR3513873,MorinelliNeebOlafsson2024}.

\begin{example}
For the symmetric Lie algebra $(\mathfrak{so}(p+d,q),\theta_{p,d,q})$, the dual symmetric Lie algebra is isomorphic to $\mathfrak{so}(p,d+q)$. More explicitly,
\begin{subequations}
\begin{align}
\mathfrak{so}(p+d,q)^{\theta_{p,d,q}}
&=
\left\{\begin{pmatrix}
X_{11}&0&X_{13}\\
0&X_{22}&0\\
X_{13}^{t}&0&X_{33}
\end{pmatrix}\middle|
\begin{array}{l}
X_{11}\in \mathfrak{gl}_{p}(\R), X_{22} \in \mathfrak{gl}_{d}(\R)\\
X_{33}\in \mathfrak{gl}_{q}(\R), X_{13}\in M_{p\times q}(\R),\\
 X_{ii}=-X_{ii}^t
\end{array}
\right\}\\
\mathfrak{so}(p+d,q)^{-\theta_{p,d,q}}
&=
\left\{\begin{pmatrix}
0&X_{12}&0\\
-X_{12}^{t}&0&X_{23}\\
0&X_{23}^{t}&0
\end{pmatrix}\middle| \begin{array}{l} X_{12}\in M_{p\times d}(\R),\\
 X_{23}\in M_{d\times q}(\R) 
 \end{array}
\right\}.
\end{align}
\end{subequations}
An isomorphism
\begin{equation}
\mathfrak{so}(p+d,q)^{*}\longrightarrow \mathfrak{so}(p,d+q)
\end{equation}
is obtained by the natural scalar multiplication in the complexification followed by conjugation by the diagonal matrix
\begin{equation}
J_{p,d,q}^{1/2}:=\diag(\underbrace{1,\ldots,1}_{p\text{ times}},\underbrace{i,\ldots,i}_{d\text{ times}},\underbrace{1,\ldots,1}_{q\text{ times}}).
\end{equation}
Explicitly, it sends $\alpha\otimes X\in \mathfrak{so}(p+d,q)^{*}$ to $\alpha\,\Ad(J_{p,d,q}^{1/2})(X)$.
\end{example}

\subsection{The dual contraction}
Let $(\fg,\theta)$ be a finite-dimensional real symmetric Lie algebra. The In\"on\"u-Wigner contraction of $\fg$ with respect to $\fg^{\theta}$ is realized by the family of linear maps
\begin{equation}
T_{\epsilon}^{\fg,\theta}(X+Y)=X+\epsilon Y,
\qquad X\in \fg^{\theta},\ Y\in \fg^{-\theta},\ \epsilon>0.
\end{equation}
Its limit is the Lie algebra
\begin{equation}
\fg^{\theta}\ltimes \fg^{-\theta}.
\end{equation}

\begin{definition}
Let $(\fg,\theta)$ be a finite-dimensional real symmetric Lie algebra. The \emph{dual contraction} of the In\"on\"u-Wigner contraction of $\fg$ with respect to $\fg^{\theta}$ is the In\"on\"u-Wigner contraction of the dual symmetric Lie algebra $\fg^{*}$ with respect to the common subalgebra $\R\otimes_{\R}\fg^{\theta}$. It is realized by the maps
\begin{equation}
T_{\epsilon}^{\fg^{*},\theta}(r_{1}\otimes X+ir_{2}\otimes Y)=r_{1}\otimes X+i\epsilon r_{2}\otimes Y,
\end{equation}
for $X\in \fg^{\theta}$, $Y\in \fg^{-\theta}$, $r_{1},r_{2}\in \R$, and $\epsilon>0$.
\end{definition}
With this notation, the limiting Lie algebra of the dual contraction is
\begin{equation}
\R\otimes_{\R}\fg^{\theta}\ltimes (i\R\otimes_{\R}\fg^{-\theta}) \cong  \fg^{\theta}\ltimes \fg^{-\theta}.
\end{equation}
\section{From dual contractions to algebraic families}\label{sec:main}
The following theorem makes precise the relation between a symmetric Lie algebra, its In\"on\"u-Wigner contraction, and the corresponding dual contraction.

\begin{theorem}\label{thm:main-jmp}
Let $(\fg,\theta)$ be a finite-dimensional real symmetric Lie algebra. Then there exist an algebraic family of complex Lie algebras $\fgfam$ over $\A^{1}_{\C}$ and an anti-holomorphic involution $\boldsymbol{\sigma}:\fgfam\to \fgfam$ such that, for every $\alpha\in \R$,
\begin{equation}\label{57}
\fgfam^{\boldsymbol{\sigma}}|_{\alpha}\cong
\begin{cases}
\fg^{*}, & \alpha<0,\\
\fg^{\theta}\ltimes \fg^{-\theta}, & \alpha=0,\\
\fg, & \alpha>0.
\end{cases}
\end{equation}
\end{theorem}

\begin{proof}
Let $\sigma$ be the anti-holomorphic involution of $\fg(\C)$ associated with the real form $\fg$, and let $\sigma^{*}=\sigma\widetilde{\theta}$ be the anti-holomorphic involution associated with the dual real form $\fg^{*}$. By Lemma~\ref{lem:dual-symmetric-jmp}, the involutions $\widetilde{\theta}$, $\sigma$, and $\sigma^{*}$ commute. The existence of $\g$ with its anti-holomorphic involution $\boldsymbol{\sigma}$ satisfying Eq.  (\ref{57}) follows  from Thm. 3.1. of  Ref.~\cite{MR3797197}.
\end{proof}

In other words, every real symmetric Lie algebra determines a single algebraic family of real Lie algebras that contains both dual real forms and their common In\"on\"u-Wigner contraction.

\subsection{Explicit realization of the family}\label{subsec:explicit}
We now recall the explicit realization of the family $\fgfam$ used in the proof of Theorem~\ref{thm:main-jmp}, following Ref.~\cite{MR3797197}. By Ado's theorem there is an integer $n\in \N$ and an embedding $\fg(\C)\hookrightarrow \mathfrak{gl}_{n}(\C)$. By Lemma~\ref{lem:real-form-embedding-jmp}, after replacing this embedding if necessary, we may assume that $\sigma$ is given by entrywise complex conjugation on matrices.

Define
\begin{equation}
\iota:\fg(\C)\longrightarrow \mathfrak{gl}_{2n}(\C),
\qquad
\iota(X)=\frac{1}{2}
\begin{pmatrix}
X+\widetilde{\theta}(X) & X-\widetilde{\theta}(X)\\
X-\widetilde{\theta}(X) & X+\widetilde{\theta}(X)
\end{pmatrix}.
\end{equation}
Then $\iota$ is an embedding of complex Lie algebras satisfying, for every $X\in \fg(\C)$,
\begin{enumerate}
    \item $\iota(\sigma(X))=\overline{\iota(X)}$,
    \item $\iota(\widetilde{\theta}(X))=J_{n,n}\,\iota(X)\,J_{n,n}$,
    \item $\iota(\sigma^{*}(X))=J_{n,n}\,\overline{\iota(X)}\,J_{n,n}$,
\end{enumerate}
where $J_{n,n}=\diag(I_{n},-I_{n})$.

In this realization,
\begin{equation}
\iota(\fg(\C))=
\left\{
\begin{pmatrix}
X_+ & X_-\\
X_- & X_+
\end{pmatrix}\middle|
X_{\pm}\in \fg(\C)^{\pm\widetilde{\theta}}
\right\}.
\end{equation}
The family $\fgfam$ is the $\C[z]$-submodule of the constant family $\mathfrak{gl}_{2n}(\C[z])$ spanned by
\begin{equation}
\left\{
\begin{pmatrix}
X_+ & X_-\\
zX_- & X_+
\end{pmatrix}\middle|
X_{\pm}\in \fg(\C)^{\pm\widetilde{\theta}}
\right\}.
\end{equation}
The anti-holomorphic involution $\boldsymbol{\sigma}:\fgfam\to \fgfam$ is given by entrywise complex conjugation on coefficients:
\begin{equation}
\boldsymbol{\sigma}(A)_{ij}=\overline{A_{ij}},
\end{equation}
where complex conjugation on $\C[z]$ was defined in Section \ref{sec:families}.
Therefore, for $\alpha\in \R$, the fiber $\fgfam^{\boldsymbol{\sigma}}|_{\alpha}$ is isomorphic to 
\begin{equation}
\left\{
\begin{pmatrix}
X_+ & X_-\\
\alpha X_- & X_=
\end{pmatrix}\middle|
X_{\pm}\in \fg(\C)^{\pm\widetilde{\theta}}\cap \fg(\C)^{\sigma}
\right\}
\cong
\begin{cases}
\fg^{*}, & \alpha<0,\\
\fg^{\theta}\ltimes \fg^{-\theta}, & \alpha=0,\\
\fg, & \alpha>0.
\end{cases}
\end{equation}

\begin{example}
For the symmetric Lie algebra $(\mathfrak{so}(p+d,q),\theta_{p,d,q})$, the fiber $\fgfam^{\boldsymbol{\sigma}}|_{\alpha}$ can be written explicitly as
\begin{equation}
\left\{\begin{pNiceArray}{ccc|ccc}
X_{11}& 0 & X_{13}  &0& X_{12} & 0\\
0& X_{22} & 0 & -X_{12}^{t}& 0& X_{23}\\
X_{13}^{t}& 0 & X_{33}&0& X_{23}^{t} &0\\
\hline
0& \alpha X_{12} & 0& X_{11}& 0 & X_{13}\\
-\alpha X_{12}^{t}& 0& \alpha X_{23} & 0& X_{22} & 0\\
0& \alpha X_{23}^{t} &0 & X_{13}^{t}& 0 & X_{33}
\end{pNiceArray}\middle|
\begin{array}{l}
X_{11}\in \mathfrak{gl}_{p}(\R), X_{22}\in \mathfrak{gl}_{d}(\R),\\
 X_{33}\in \mathfrak{gl}_{q}(\R), X_{13}\in M_{p\times q}(\R)\\
X_{12}\in M_{p\times d}(\R),\\   X_{23}\in M_{d\times q}(\R)\\
X_{ii}^{t}=-X_{ii}
\end{array}
\right\}.
\end{equation}
For $\alpha>0$ this fiber is isomorphic to $\mathfrak{so}(p+d,q)$, for $\alpha=0$ it is the contraction $(\mathfrak{so}(p,q)\oplus \mathfrak{so}(d))\ltimes (M_{p\times d}(\R)\oplus M_{d\times q}(\R))$, and for $\alpha<0$ it is isomorphic to $\mathfrak{so}(p,d+q)$.
\end{example}

\section{Discussion}
We introduced the notion of a dual contraction for the In\"on\"u-Wigner contraction of a real Lie algebra with respect to a symmetric subalgebra. The main theorem shows that the original contraction and its dual occur as real fibers of a single algebraic family. In this way, the two contractions are not merely related by analogy: they are part of one algebraic and analytic object.

This point of view suggests that information can sometimes be transferred from one side of the family to the other. For instance, in the hydrogen-atom setting studied in Refs.~\cite{MR3827131,MR4460278}, the relevant family is precisely the one attached to the dual contractions
\begin{equation}
\mathfrak{so}(n+1)\longrightarrow \mathfrak{so}(n)\ltimes \R^{n}
\qquad \text{and} \qquad
\mathfrak{so}(n,1)\longrightarrow \mathfrak{so}(n)\ltimes \R^{n}.
\end{equation}
Those works show that algebraic families of Harish-Chandra modules that arise as algebraic solutions for the Schr\"odinger  equation    encode spectral information for the Schr\"odinger operator.
Moreover, by analyticity, it is enough to know the algebraic solutions for positive energies in order to determine the complete spectrum.

\section{Appendix}\label{sec:appendix}

\begin{proposition}\label{prop:equivalent-subalgebras-jmp}
Let $\fk,\fk'\in \Sub(\fg)$. Suppose there is a vector-space isomorphism $\nu:\fg\to \fg$ whose restriction $\nu|_{\fk}:\fk\to \fk'$ is an isomorphism of Lie algebras and such that
\begin{equation}
\nu([X,Y])-[\nu(X),\nu(Y)]\in \fk',
\qquad \forall X\in \fk,\ Y\in \fg.
\end{equation}
Then $\fk\sim \fk'$.
\end{proposition}

\begin{proof}
By Lemma~\ref{lem:complements-jmp}, both contractions associated with $\fk$ and $\fk'$ can be written canonically as semidirect products with quotients:
\begin{equation}
\fg_{0}^{\fk,\fp}\cong \fk\ltimes (\fg/\fk),
\qquad
\fg_{0}^{\fk',\fp'}\cong \fk'\ltimes (\fg/\fk').
\end{equation}
Hence it suffices to construct an isomorphism
\begin{equation}
\widetilde{\nu}:\fk\ltimes (\fg/\fk)\longrightarrow \fk'\ltimes (\fg/\fk').
\end{equation}
Define
\begin{equation}
\widetilde{\nu}(X,Y+\fk)=(\nu(X),\nu(Y)+\fk'),
\qquad X\in \fk,\ Y\in \fg.
\end{equation}
This is clearly a vector-space isomorphism. It preserves the brackets on the subalgebra part and on the abelian ideal part, so it remains only to check mixed brackets. For $X\in \fk$ and $Y\in \fg$,
\begin{align}
\widetilde{\nu}([(X,0+\fk),(0,Y+\fk)])
&=\widetilde{\nu}(0,[X,Y]+\fk)\nonumber\\
&=(0,\nu([X,Y])+\fk').
\end{align}
By the hypothesis,
\begin{equation}
\nu([X,Y])+\fk'=[\nu(X),\nu(Y)]+\fk',
\end{equation}
and therefore
\begin{equation}
\widetilde{\nu}([(X,0+\fk),(0,Y+\fk)])
=[(\nu(X),0+\fk'),(0,\nu(Y)+\fk')]
=[\widetilde{\nu}(X,0+\fk),\widetilde{\nu}(0,Y+\fk)].
\end{equation}
Hence $\widetilde{\nu}$ is an isomorphism of Lie algebras, so $\fk\sim \fk'$.
\end{proof}

\begin{lemma}\label{lem:real-form-embedding-jmp}
Let $\fl$ be a finite-dimensional complex Lie algebra and let $\sigma:\fl\to \fl$ be an anti-holomorphic involution. Then there exist an integer $m\in \N$ and an embedding
\begin{equation}
\phi:\fl\longrightarrow \mathfrak{gl}_{2m}(\C)
\end{equation}
of complex Lie algebras such that
\begin{equation}
\phi(\sigma(X))=\overline{\phi(X)},
\qquad X\in \fl.
\end{equation}
\end{lemma}

\begin{proof}
By Ref.~\cite[Lemma~3.1]{MR3797197} there exist $m\in \N$ and an embedding
\begin{equation}
\tau:\fl\longrightarrow \mathfrak{gl}_{2m}(\C)
\end{equation}
of complex Lie algebras such that
\begin{equation}
\tau(\sigma(X))=S\,\overline{\tau(X)}\,S,
\end{equation}
where
\begin{equation}
S=
\begin{pNiceArray}{c|c}
O_{m} & I_{m}\\
\hline
I_{m} & O_{m}
\end{pNiceArray}
\in \GL(2m,\R).
\end{equation}
Now let
\begin{equation}
P=
\begin{pNiceArray}{c|c}
I_{m} & I_{m}\\
\hline
iI_{m} & -iI_{m}
\end{pNiceArray}
\in \GL(2m,\C)
\end{equation}
and define
\begin{equation}
\phi(X)=P\tau(X)P^{-1}.
\end{equation}
A direct computation shows that
\begin{equation}
\phi(\sigma(X))=\overline{\phi(X)},
\qquad X\in \fl.
\end{equation}
\end{proof}

\printbibliography
\end{document}